%% file: VConPLG.tex
\documentclass[a4paper,11pt]{article}

\input{mypreambel}

\input{newcommands}

\begin{document}

\input{Titel}
\input{Intro}
\input{Prelim}
\input{Approx}

 \input{VConPLG.bbl}


\end{document}

%% file: mypreambel.tex
\usepackage[utf8]{inputenc}

\usepackage{pgf}
\usepackage{nicefrac}
\usepackage[ruled,vlined]{algorithm2e}[2005/10/04]
\usepackage{tikz,fullpage}
\usetikzlibrary{plothandlers,arrows,topaths,trees,shapes,fit,decorations.pathreplacing,decorations.text,decorations.pathmorphing,decorations.shapes,automata,positioning,chains,patterns,backgrounds,spy}

\usepackage[]{asymptote}
\usepackage[
	english,
]{babel}

\usepackage[%
]{graphicx}

\usepackage[
   centertags, 
   sumlimits,  
   intlimits,  
   namelimits, 
   fleqn,     
]{amsmath} %

\usepackage{amsthm}
\usepackage{amssymb}
\usepackage{mathtools}
\usepackage{pifont,amsfonts,amstext}

\usepackage[T1]{fontenc} 
\usepackage{textcomp}

\usepackage{lmodern}

\usepackage[%
	english
]{nomencl}

\definecolor{sectioncolor}{RGB}{0, 0, 0}    
\definecolor{textcolor}{RGB}{0, 0, 0}        
\definecolor{shadecolor}{gray}{0.90}
\definecolor{pdfurlcolor}{rgb}{0,0,0.6}
\definecolor{pdffilecolor}{rgb}{0.7,0,0}
\definecolor{pdflinkcolor}{rgb}{0,0,0.6}
\definecolor{pdfcitecolor}{rgb}{0,0,0.6}
\colorlet{stringcolor}{green!40!black!100}
\colorlet{commencolor}{blue!0!black!100}

\usepackage{aliascnt}

\usepackage[
  colorlinks=true,         
  urlcolor=pdfurlcolor,    
  filecolor=pdffilecolor,  
  linkcolor=pdflinkcolor,  
  citecolor=pdfcitecolor,  %
  raiselinks=true,			 
  breaklinks,              
  verbose,
  hyperindex=true,         
  linktocpage=true,        
  hyperfootnotes=false,     
  bookmarks=true,          
  bookmarksopenlevel=1,    
  bookmarksopen=true,      
  bookmarksnumbered=true,  
  bookmarkstype=toc,       
  plainpages=false,        
  pageanchor=true,         
]{hyperref}

\let\orgautoref\autoref
\providecommand{\Autoref}
        {\def\equationautorefname{Equation}%
         \def\figureautorefname{Figure}%
         \def\subfigureautorefname{Figure}%
         \def\Itemautorefname{Item}%
         \def\tableautorefname{Table}%
         \def\sectionautorefname{Section}%
         \def\subsectionautorefname{Section}%
         \def\subsubsectionautorefname{Section}%
         \def\chapterautorefname{Section}%
         \def\partautorefname{Part}%
	 \def\algocfautorefname{Algorithm}
         \orgautoref}

\def\subsectionautorefname{Section}%
\def\subsubsectionautorefname{Section}%
\def\algocfautorefname{Algorithm}

\theoremstyle{plain}
 
\newtheorem{theorem}{Theorem}  
 
\newaliascnt{lemma}{theorem}  
\newtheorem{lemma}[lemma]{Lemma}  
\aliascntresetthe{lemma}

\newtheorem*{theorem*}{Theorem}

\theoremstyle{definition}

\theoremstyle{remark}

\theoremstyle{definition}
\usepackage{shadethm}
\newshadetheorem{probs}{Problem}
\newshadetheorem{defs}{Definition}

%% file: newcommands.tex
\newcommand{\N}{\mathcal{N}}

\DeclareMathOperator{\e}{e}

\newcommand{\OPT}{\textsf{OPT}}

\newcommand{\PLG}{\textsc{PLG}}

\newcommand{\MVC}{\textsc{Min-VC}}
\newcommand{\MIS}{\textsc{Max-IS}}
\newcommand{\MDS}{\textsc{Min-DS}}

\newcommand{\UGC}{\textsc{UGC}}
\newcommand{\ugc}{\textrm{Unique Games Conjecture}}

\newcommand{\mvc}{\textsc{Minimum Vertex Cover}}
\newcommand{\mis}{\textsc{Maximum Independent Set}}
\newcommand{\mds}{\textsc{Minimum Dominating Set}}

\newcommand{\LP}{\textsf{LP}}

\newcommand{\NP}{\textsc{NP}}
\renewcommand{\P}{\textsc{P}}
\newcommand{\APX}{\textsc{APX}}

\newcommand{\E}{\mathbb{E}}

\newcommand{\weight}{\mathrm{w}}

\renewcommand{\deg}[1]{\ensuremath{d(#1)}}
\newcommand{\gab}{\ensuremath{G}}
\newcommand{\Gab}{\ensuremath{\mathcal{G}_{(\alpha,\beta)}}}
\newcommand{\Mab}{\ensuremath{M(\alpha,\beta)}}
\newcommand{\dmax}{\ensuremath{\Delta}}
\newcommand{\Vgood}{\ensuremath{V^{*}}}

\newcommand{\mainratio}{\ensuremath{2-\frac{\zeta(\beta)-1-\frac{1}{2^{\beta}}}{2^{\beta}\zeta(\beta-1)\zeta(\beta)}}}
\newcommand{\newratio}{\ensuremath{2-\frac{\left(\zeta(\beta)-1-\frac{1}{2^{\beta}}\right) \cdot \zeta(\beta-1)}{\zeta(\beta-1) \cdot \zeta(\beta)}
 \left[1-{\left(\frac{\zeta(\beta-1)-\left(1+\frac{1}{2^{\beta-1}}\right)}{\zeta(\beta-1)}\right)}^{3}\right]}}
\newcommand{\etau}{\ensuremath{\frac{\sum_{i = 1}^{\dmax}\frac{e^{\alpha}}{i^{\beta-1}}-e^{\frac{\alpha}{\beta}}+1}{\sum_{i = 1}^{\dmax}\frac{e^{\alpha}}{i^{\beta-1}}}\left[1-\left(\frac{\sum_{i = 1}^{\dmax}\frac{e^{\alpha}}{i^{\beta-1}}-\deg{U}-3+1}{\sum_{i = 1}^{\dmax}\frac{e^{\alpha}}{i^{\beta-1}}-3+1}\right)^{3}\right]}}

\def\clap#1{\hbox to 0pt{\hss#1\hss}}
\def\mathclap{\mathpalette\mathclapinternal}
\def\mathclapinternal#1#2{%
\clap{$\mathsurround=0pt#1{#2}$}%
}

\newcommand{\unproc}{\upshape\textrm{\texttt{unprocessed}}}
\newcommand{\proc}{\upshape\textrm{\texttt{processed}}}
\newcommand{\comp}{\texttt{compute}}
\newcommand{\aset}{\texttt{set}}

\newcommand{\astep}[1]{{\scriptsize\textrm{\textbf{(#1)}}}}

\newcommand{\floor}[1]{\left\lfloor {#1} \right\rfloor}
\newcommand{\ceil}[1]{\left\lceil {#1} \right\rceil}

\renewcommand{\leq}{\leqslant}
\renewcommand{\geq}{\geqslant}
\newcommand{\restrict}{\!\restriction\!}

\tikzset{vertex/.style={draw,shading=ball,ball color=black,circle, inner sep=0pt, minimum size=4pt}}
\tikzset{vertex2/.style={draw,shading=ball,ball color=black!25,circle,  inner sep=0pt, minimum size=4pt}}
\tikzset{every label/.style={label distance=-3pt}, node distance = 12pt}

\newcommand*{\triplet}[4]{\tikz[baseline=0pt] {
  \node[#1,label=above:\tiny$#2$] (u) {};
  \node[vertex,label=above:\tiny$\nicefrac{1}{2}$,right= of u] (v) {};
  \node[#3,label=above:\tiny$#4$,right= of v] (w) {};
  \path[draw] (u) -- (v) -- (w);
  \node[right=5pt of w, yshift=2pt] (mapsto) {$\mapsto$};
  \node[vertex2,label=above:\tiny$1$,right=5pt of mapsto, yshift=-2pt] (u) {};
  \node[vertex2,label=above:\tiny$0$,right= of u] (v) {};
  \node[vertex2,label=above:\tiny$1$,right= of v] (w) {};
  \path[draw] (u) -- (v) -- (w);
 }}

\newcommand*{\zweisprung}{\tikz[baseline=0pt] {
  \node[vertex,label=above:\tiny$\nicefrac{1}{2}$] (u) {};
  \node[vertex,label=above:\tiny$\nicefrac{1}{2}$,right= of u] (v) {};
  \node[vertex,label=above:\tiny$\nicefrac{1}{2}$,right= of v] (w) {};
  \path[draw] (u) -- (v) -- (w);
  \node[right=5pt of w, yshift=2pt] (mapsto) {$\mapsto$};
  \node[vertex,label=above:\tiny$\nicefrac{1}{2}$,right=5pt of mapsto, yshift=-2pt] (u) {};
  \node[vertex2,label=above:\tiny$1$,right= of u] (v) {};
  \node[vertex2,label=above:\tiny$0$,right= of v] (w) {};
  \path[draw] (u) -- (v) -- (w);
  \node[right=5pt of w, yshift=2pt] (mapsto) {$\mapsto$};
  \node[vertex2,label=above:\tiny$1$,right=5pt of mapsto, yshift=-2pt] (u) {};
  \node[vertex2,label=above:\tiny$1$,right= of u] (v) {};
  \node[vertex2,label=above:\tiny$0$,right= of v] (w) {};
  \path[draw] (u) -- (v) -- (w);
 }}

\newcommand*{\doublet}{\tikz[baseline=0pt] {
  \node[vertex,label=above:\tiny$\nicefrac{1}{2}$] (u) {};
  \node[vertex,label=above:\tiny$1$,right= of u] (v) {};
  \path[draw] (u) -- (v);
  \node[right=5pt of v, yshift=2pt] (mapsto) {$\mapsto$};
  \node[vertex2,label=above:\tiny$1$,right=5pt of mapsto, yshift=-2pt] (u) {};
  \node[vertex2,label=above:\tiny$1$,right= of u] (v) {};
  \path[draw] (u) -- (v);
 }}

\newcommand*{\quartet}{\tikz[baseline=0pt] {
  \node[vertex,label=above:\tiny$\geq\negthickspace\nicefrac{1}{2}$] (u) {};
  \node[vertex,label=above:\tiny$\nicefrac{1}{2}$,right= of u] (v1) {};
  \node[vertex,label=above:\tiny$\geq\negthickspace\nicefrac{1}{2}$,right= of v1] (v2) {};
  \node[vertex,label=above:\tiny$\geq\negthickspace 0$,right= of v2] (w) {};
  \path[draw] (u) -- (v1) -- (v2) -- (w);
  \node[right= of w, yshift=2pt] (mapsto) {$\mapsto$};
  \node[vertex2,label=above:\tiny$1$,right= of mapsto, yshift=-2pt] (u) {};
  \node[vertex2,label=above:\tiny$1$,right= of u] (v1) {};
  \node[vertex2,label=above:\tiny$0$,right= of v1] (v2) {};
  \node[vertex2,label=above:\tiny$1$,right= of v2] (w) {};
  \path[draw] (u) -- (v1) -- (v2) -- (w);
 }}

\newcommand*{\striplet}{\tikz[baseline=-5pt] {
  \node[vertex,label=below:\tiny$u$] (u) {};
  \node[vertex,label=below:\tiny$v$,right= of u] (v) {};
  \node[vertex,label=below:\tiny$w$,right= of v] (w) {};
  \path[draw] (u) -- (v) -- (w);
 }}

\newcommand*{\squartet}{\tikz[baseline=-5pt] {
  \node[vertex,label=below:\tiny$u$] (u) {};
  \node[vertex,label=below:\tiny$v_1$,right= of u] (v1) {};
  \node[vertex,label=below:\tiny$v_2$,right= of v1] (v2) {};
  \node[vertex,label=below:\tiny$w$,right= of v2] (w) {};
  \path[draw] (u) -- (v1) -- (v2) -- (w);
 }}

%% file: Titel.tex
\title{\bf Approximability of the Vertex Cover Problem in Power Law Graphs}
\author{Mikael Gast\thanks{Dept. of Computer Science, University of Bonn.
    e-mail:{ \tt gast@cs.uni-bonn.de}} \and
	Mathias Hauptmann\thanks{Dept. of Computer Science, University of Bonn.
    e-mail:{ \tt hauptman@cs.uni-bonn.de}}}
\date{}
\maketitle

\begin{abstract}
In this paper we construct an approximation algorithm for the \textsc{Minimum Vertex Cover} Problem (\textsc{Min-VC}) with an expected approximation ratio of $2-\frac{\zeta(\beta)-1-\frac{1}{2^{\beta}}}{2^{\beta}\zeta(\beta-1)\zeta(\beta)}$ for random \emph{Power Law Graphs} (PLG) in the $(\alpha,\beta)$-model of Aiello et. al.. We obtain this result by combining the Nemhauser and Trotter approach for \textsc{Min-VC} with a new deterministic rounding procedure which achieves an approximation ratio of $\frac{3}{2}$ on a subset of low degree vertices for which the expected contribution to the cost of the associated linear program is sufficiently large.
\end{abstract}

%% file: Intro.tex
\section{Introduction}

In recent years topological analyses have been applied to a variety of real world graphs such as the World-Wide Web, the Internet, Collaboration and Social Networks, Protein Interaction Networks and other large-scale graphs of biological systems.
Typical statistical parameters such as the \emph{diameter}, \emph{robustness}, \emph{clustering coefficient} and \emph{degree distribution} have been measured and compared to the expected values of these parameters in \emph{uniform} random graph models such as the classical $G(n,p)$-Model due to Erd\H{o}s and Rényi \cite{Erdos1960}. It turned out that the real world graphs are significantly different from the random models with respect to these statistical and topological properties. In subsequent studies the aim was to describe the properties of real world networks mathematically and to propose new models in order to meet these conditions.

As of 1999 Kumar et. al. \cite{Broder2000,Kumar2000}, Kleinberg et. al. \cite{Kleinberg1999a,Kleinberg2001} and Faloutsos, Faloutsos and Faloutsos \cite{Faloutsos1999,Siganos2003} measured the degree sequence of the World-Wide Web and independently observed that it is well approximated by a power law distribution, i.e. the number of nodes $y_i$ of a given degree $i$ is proportional to $i^{-\beta}$ where $\beta > 0$.
This was later verified for a large number of existing real-world networks such as protein-protein interactions, gene regulatory networks, peer-to-peer networks, mobile call networks and social networks \cite{Jovanovic2001,Guelzim2002,Seshadri2008,Eubank2004}.

In order to analyze these graphs, some research has been directed towards finding suitable models for describing structural properties quantitatively and qualitatively. A number of \emph{Power Law Graph} (\PLG) models have been proposed, such as the Barab\'asi-Albert model of \emph{Preferential Attachment} \cite{Barabasi1999}, the \emph{Buckley-Osthus Model} \cite{Buckley2004}, the \emph{Cooper-Frieze Model} \cite{Cooper2003} and the \emph{Copying Model} due to Kumar et. al. \cite{Kumar2000}. All these models describe a \emph{random growth process} starting from a small seed graph and yielding -- besides other features -- a power law degree sequences in the limit.

A different approach is to take a power law degree sequence as input and to generate a graph instance with this distribution in a random fashion. Among the most widely known models of this kind is the \emph{ACL-Model} due to Aiello, Chung and Lu \cite{Aiello2001}.
Here, the number $y_i$ of vertices with degree $i$ is roughly given by $y_i \approx e^{\alpha} / i^{\beta}$, where $e^{\alpha}$ is a normalization constant which determines the size of the graph.
While this model is potentially less accurate than the detailed description of a growth process, it has the advantage of being robust and general, i.e., structural properties that are true in this model will be true for the majority of graphs with the given degree sequence.

All of the above models are well motivated and there exists a large body of literature on mathematical foundations and applications \cite{Barabasi1999,Aiello2000,Bollobas2002,Eubank2004,Mihail2006}. In this paper, we focus on the ACL-Model for random \PLG\ which we will refer to as the \emph{$(\alpha,\beta)$-Model}.

Apart from having certain structural properties, such as high clustering coefficient, small-world characteristics and self similarity, there exists practical evidence that combinatorial optimization in \PLG\ is easier than in general graphs \cite{Park2001,Gkantsidis2003,Eubank2004,Koyuturk2006}.
Contrasting this Ferrante et. al. \cite{Ferrante2008} and Shen et. al. \cite{Shen2010} studied the approximation hardness of certain optimization problems in \emph{combinatorial Power Law Graphs} and showed \NP-hardness and \APX-hardness of classical problems such as \mvc\ (\MVC), \mis\ (\MIS) and \mds\ (\MDS).
In this paper we study the approximability of the \mvc\ problem in the random Power Law Graph model of Aiello et. al. \cite{Aiello2001}.

The \mvc\ is one of the most well-studied problems in combinatorial optimization. A \emph{vertex cover} of a graph $G=(V,E)$ is a set of vertices $C\subseteq V$ such that each edge  $e=\{u,v\}$ of $G$ has at least one endpoint in $C$. The \mvc\ problem (\MVC) is the the problem of finding a cover of minimum cardinality in a graph.
The problem is known to be \NP-complete due to Karp's original proof \cite{Karp1972} and \APX-complete \cite{Papadimitriou1991}. Moreover, it cannot be approximated within a factor of $1.3606$ \cite{Dinur2005}, unless $\P = \NP$, and is inapproximable within $2-\epsilon$ for any $\epsilon > 0$ as long as the \ugc\ (\UGC) holds true \cite{Khot2008}.
Here, we show that the \MVC\ problem can be approximated with an expected approximation ratio $<2$ in random Power Law Graphs:

\begin{theorem}\label{thm:main}
 There exists a polynomial time algorithm which approximates the \mvc\ problem (\MVC) in random Power Law Graphs in the $(\alpha,\beta)$-Model for $\beta > 2$ (where graphs are given instance by instance) with an expected approximation ratio of \[\rho=\mainratio.\]
\end{theorem}

We also give a refined analysis for the case $\beta > 2.424$ and obtain the following improvement.

\begin{theorem}\label{thm:main2}
 For $\beta>2.424$, the \mvc\ problem (\MVC) in the $(\alpha,\beta)$-Model can be approximated with expected asymptotic approximation ratio \[\rho'=\newratio.\]
\end{theorem}
In \autoref{fig:Comparison} these two upper bounds $\rho$ and $\rho'$ are shown as functions of the parameter $\beta$.
\begin{figure}[htp]
 \centering
%
%
%
%
 \includegraphics[]{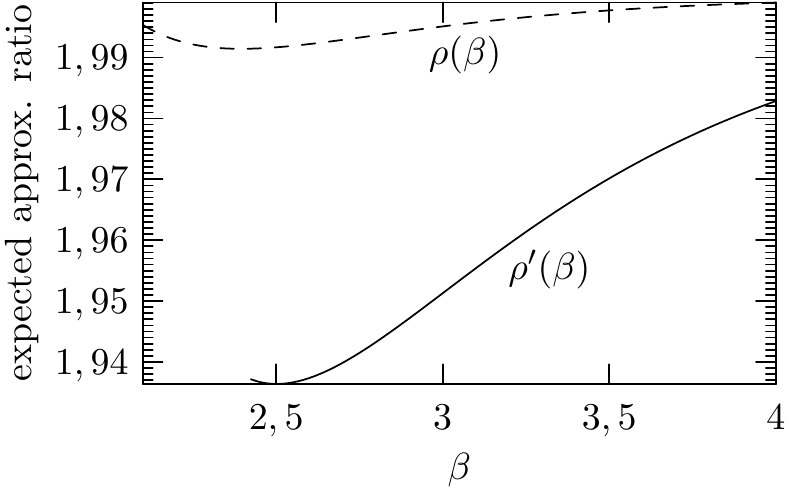}
 \caption{Comparison of first (\protect\tikz[baseline=-3pt]{\protect\path (0,0) edge[dashed,bend left= 10] (.5,0);}) and second (\protect\tikz[baseline=-3pt]{\protect\path (0,0) edge[bend left= 10] (.5,0);}) analysis in terms of functions of the parameter $\beta$, for $\beta > 2$ and $\beta > 2.424$, respectively.}
 \label{fig:Comparison}
\end{figure}

The paper is organized as follows.
In \autoref{subsec:alpha_beta} we describe the $(\alpha,\beta)$-model for Power Law Graphs, describe the random generation process and give a formal description of the model parameters.
In \autoref{subsec:vertex_cover} we give some background on the \MVC\ problem and briefly describe the \emph{half-integral solution} method proposed by Nemhauser and Trotter.
\Autoref{sec:vc_on_plg} presents our new approximation algorithm for \MVC\ in Power Law Graphs. This algorithm basically consists of a deterministic rounding procedure on a half-integral solution for \MVC.
In \autoref{sec:algorithm} we show that this rounding procedure yields an approximation ratio of $\frac{3}{2}$ in the subgraph induced by the low-degree vertices of the Power Law Graph and a 2-approximation in the residual graph.
In \autoref{subsec:bounds} we construct upper and lower bounds on the expected size of the half-integral solution in the induced subgraph of low-degree vertices and finally prove our main theorems.
We conclude the paper by giving a short summary and further research in \autoref{sec:conclusion}.

%% file: Prelim.tex
\section{Preliminaries}

\subsection{\texorpdfstring{$(\alpha,\beta)$}--Power Law Graphs}\label{subsec:alpha_beta}

In this section we describe the random \PLG-model proposed by Aiello, Chung and Lu in \cite{Aiello2001}, which we will denote as \Mab. This model considers a random graph with the following degree distribution depending on two given values $\alpha$ and $\beta$: 
For each $1\leq i \leq \Delta = \floor{e^{\frac{\alpha}{\beta}}}$ there are $y_i$ vertices of degree $i$ with 
\[
y_i=
\begin{cases}
  \floor{\frac{e^{\alpha}}{i^{\beta}}} & \text{if } i > 1 \text{ or } \sum^{\Delta}_{i=1} \floor{\frac{e^{\alpha}}{i^{\beta}}} \text{ is even}\\
  \floor{e^{\alpha}} + 1 & \text{otherwise.}
\end{cases}
\]
Here, $i$ and $y_i$ satisfy $\log y_i = \alpha - \beta \log i$. The variable $\alpha$ is the logarithm of the size of the graph and $\beta$ is the $\log$-$\log$ growth rate. Let \Gab\ be the set of all undirected graphs with multi-edges and self-loops on $n=\sum_{i=1}^{\Delta} y_i$ vertices which have $y_i$ vertices of degree $i$ ($1 \leq i \leq \Delta$). Then \Mab\ is the distribution on \Gab\ obtained in the following way \cite{Aiello2001}:
\begin{enumerate}
  \item Generate a set $L$ of $\deg{v}$ distinct copies of each vertex $v$.
  \item Generate a random matching on the elements of $L$.
  \item For each pair of vertices $u$ and $v$, the number of edges joining $u$ and $v$ in $\gab$ is equal to the number  of edges in the matching of $L$ which join copies of $u$ to copies of $v$.
\end{enumerate}

\begin{center}
 \input{RandomMatch.pgf}
\end{center}

As in \cite{Aiello2001}, in the following we will work with the real numbers $\frac{e^{\alpha}}{i^{\beta}}$, $e^{\frac{\alpha}{\beta}}$ instead of their integer counterparts. For $\beta > 2$ the error is a lower order term (c.f. \cite{Aiello2001}, remark on page 6). 

A graph $G \in\Gab$ has the following properties: The maximum degree of $G$ is $e^{\frac{\alpha}{\beta}}$, and for $\beta > 2$ the number of vertices is $n=\sum^{e^{\nicefrac{\alpha}{\beta}}}_{i=1} \frac{e^{\alpha}}{i^{{\beta}}} \approx \zeta(\beta)e^{\alpha}$ and the number of edges is $m=\frac{1}{2}\sum^{e^{\nicefrac{\alpha}{\beta}}}_{i=1} i \frac{e^{\alpha}}{i^{{\beta}}} \approx \frac{1}{2}\zeta(\beta-1)e^{\alpha}$ where the error terms are $o(n)$ and $o(m)$, respectively.

\subsection{LP-Relaxation and Half-Integral Solution for Min-VC}\label{subsec:vertex_cover}

In this section we give a brief outline of the \emph{Nemhauser-Trotter Theorem} stated in \cite{Nemhauser1975} and show how this is used to approximate \MVC\ in a graph $G=(V,E)$ as described by Hochbaum et. al. in \cite{Hochbaum1993}.

Nemhauser and Trotter considered the following $\LP$-relaxation, which applies to the more general \emph{weighted} vertex cover problem:
\begin{alignat*}{3}
    \text{minimize }\quad	& \sum_{i=1}^n w_i x_i, 	&&&\quad& \\
    \text{subject to }\quad	& x_i + x_j &\geq\,& 1, 	&& \text{ for each edge } \{v_i,v_j\}\in E, \\
				& x_i &\geq\,& 0,		&& \text{ for each vertex }v_i\in V,
\end{alignat*}
They show that there always exists an optimal solution $x$ for this \LP\ which is \emph{half-integral}, i.e. for all $i$, $x_i \in \left\{0,\frac{1}{2},1 \right\}$. Then they partition the set of vertices into subsets $P, Q, R \subseteq V$, such that $v_i \in P$ if $x_i=1$, $v_i \in Q$ if $x_i=\frac{1}{2}$ and $v_i \in R$ if $x_i=0$ in this solution.
They show that at least one optimal vertex cover in $G$ contains the set $P$, that each vertex in $R$ has all its neighbors in $P$ and -- moreover -- that each cover in $G$ has weight at least $\weight(P) + \frac{1}{2}\weight(Q)$. From this it follows that at least one optimal vertex cover in $G$ consists of the set $P$ and an optimal cover in the subgraph $H[Q]$ induced by $Q$.

Hochbaum et. al. \cite{Hochbaum1993} showed that an integer solution $y$ obtained by setting $y_i=1$ for all vertices $v_i \in Q \cup P$ and $y_i=0$ for all $v_i \in R$ is a 2-approximate solution for the \MVC\ problem in $G$.
Our approximation algorithm for \MVC\ in random Power Law Graphs will make use of a half-integral solution $x$ of the \LP-relaxation along with the properties described in the Nemhauser-Trotter Theorem in order to achieve an approximation ratio strictly less than 2.

%% file: RandomMatch.pgf
\usetikzlibrary{decorations.pathreplacing,shapes.arrows,positioning}

\begin{tikzpicture}
[
every label/.style={label distance=5pt},
vertex/.style={draw,shading=ball,ball color=black,circle, inner sep=2pt},
vertex2/.style={draw,shading=ball,ball color=black!25,circle, inner sep=2pt},
every edge/.style={draw,densely dotted,in=-90,out=-90},
every pin edge/.style={dotted,latex-,in=90,out=-45,shorten <=-4pt},
every pin/.style={dashed,pin distance=1cm}]

\fill[black!25,rounded corners=6pt] (-.25,.25) -- (.25,0.25) -- (.25,-1.25) -- (-0.25,-1.25) -- cycle;
\fill[black!25,rounded corners=6pt,xshift=1cm] (-0.25,0.25) -- (.25,0.25) -- (.25,-1.25) -- (-0.25,-1.25) -- cycle;

\fill[black!25,rounded corners=6pt,xshift=-1.5cm] (4.3,0.25) -- (4.7,0.25) -- (5,-1.25) -- (4,-1.25) -- cycle;
\fill[black!25,rounded corners=6pt] (4.3,0.25) -- (4.7,0.25) -- (5,-1.25) -- (4,-1.25) -- cycle;

\fill[black!25,rounded corners=6pt] (6.8,0.25) -- (7.2,0.25) -- (7.8,-1.25) -- (6.2,-1.25) -- cycle;
\fill[black!25,rounded corners=6pt,xshift=2cm] (6.8,0.25) -- (7.2,0.25) -- (7.8,-1.25) -- (6.2,-1.25) -- cycle;

\node[vertex,label={above:$v_{11}$}] (v1) at (0,0) {};
\node[vertex,label={above:$v_{12}$}] (v12) at (1,0) {};
\node at (2,-0.5) {$\dots$};
\node[vertex,label={above:$v_{21}$}] (v21) at (3,0) {};
\node[vertex,label={above:$v_{22}$}] (v22) at (4.5,0) {};
\node at (5.7,-0.5) {$\dots$};
\node[vertex,label={above:$v_{31}$}] (v31) at (7,0) {};
\node[vertex,label={above:$v_{32}$}] (v32) at (9,0) {};
\node at (10.25,-0.5) {$\dots$};


\node[vertex2] (v11) at (0,-1) {};

\node[vertex2] (v12) at (1,-1) {};

\node[vertex2] (v211) at (2.75,-1) {};
\node[vertex2] (v212) at (3.25,-1) {};

\node[vertex2] (v221) at (4.25,-1) {};
\node[vertex2] (v222) at (4.75,-1) {};

\node[vertex2] (v311) at (6.5,-1) {};
\node[vertex2] (v312) at (7,-1) {};
\node[vertex2] (v313) at (7.5,-1) {};

\node[vertex2] (v321) at (8.5,-1) {};
\node[vertex2] (v322) at (9,-1) {};
\node[vertex2] (v323) at (9.5,-1) {};

\draw[rounded corners] (v11) -- +(0,-.4) -| (v211);
\draw[rounded corners]  (v12)  -- +(0,-.5) -| (v311);
\draw[rounded corners]  (v221)  -- +(0,-.4) -| node (s) {}(v222);
\draw[rounded corners]  (v313)  -- +(0,-.4) -| node (m1) {}(v322);
\draw[rounded corners]  (v312)  -- +(0,-.5) -| node (m2) {}(v323);
\draw[rounded corners]  (v212) -- +(0,-.6) -| (v321);

\node[below left = .25 of m1] (lm) {multi-edge};
\node[below left = .25 of s] (ls) {self-loop};
\draw[-latex,dotted,rounded corners] (lm.east) ..controls +(1,0) and +(.5,-.5).. (m1.north);
\draw[-latex,dotted,rounded corners] (lm.east) ..controls +(.5,0) and +(.5,-.5).. (m2.north);
\draw[-latex,dotted,rounded corners] (ls.east) ..controls +(.5,0) and +(.5,-.5).. (s.north);





\end{tikzpicture}

%% file: Approx.tex
\section{Approximation of Min-VC in \texorpdfstring{$(\alpha,\beta)$}--PLG}\label{sec:vc_on_plg}

In this section we present our main result, namely an approximation algorithm with expected approximation ratio \mainratio\ for the \MVC\ problem in $(\alpha,\beta)$-\PLG\ for $\beta > 2$. Furthermore a refined analysis yields an improved \emph{asymptotic} approximation ratio for the case $\beta > 2.424$. 

Let us first give an outline of this algorithm. On instance $\gab \in \Gab$ the algorithm starts with a half-integral solution $x:V \to \left\{0,\frac{1}{2}, 1\right\}$ of the associated \LP\ and uses some deterministic rounding procedure to generate an integral solution $y: V \to \{0,1\}$. 
We show that for the set $\Vgood = \bigcup_{v:\deg{v}\in\{1,2\}} \left(\{v\} \cup \N(v)\right)$ of degree-1 and degree-2 nodes and their neighbors in \gab, the rounding procedure satisfies $ y(\Vgood) \leq \frac{3}{2} \cdot x(\Vgood)$ and furthermore $x(\Vgood)$ is sufficiently large (in expectation) with respect to \Mab.

\subsection{Approximation Algorithm}\label{sec:algorithm}

Now, we describe our deterministic rounding procedure (\Autoref{alg:rounding}) on $\gab=(V,E)$ for $\gab\in\Gab$. First, the algorithm processes all nodes of the subset $V'=L \cup \N(L)$ where $L=\left\{v \in V \vert \left(\deg{v}=2, x(v)=\frac{1}{2}\right)\vee\left(\deg{v}=1\right)\right\}$ and provides a rounded integral solution $y$ with $y(V')\leq\frac{3}{2}\cdot x(V')$. Furthermore we show that $y(\Vgood\setminus V') \leq\frac{4}{3}\cdot x(\Vgood\setminus V')$ and $y(V\setminus\Vgood)\leq 2\cdot x(V\setminus\Vgood)$.

\begin{algorithm}
\SetNlSty{textbf}{(}{)}
\KwIn{$\gab=(V,E) , x:V \to \left\{0,\frac{1}{2},1\right\}$.}
\KwOut{$y:V \to \{0,1\}$.}
\BlankLine
\ForAll{$v \in V$}{
 $y(v) := x(v)$\;
 mark $v$ as \unproc\;
}
\nlset{0}$\comp$ $G'=(V',E')$ induced by $V'=L\cup\N(L)$ where $L=\left\{v \in V \vert \left(\deg{v}=2, x(v)=\frac{1}{2}\right)\vee\left(\deg{v}=1\right)\right\}$\;
\nlset{1}\ForAll{$v \in V$ with $\deg{v}=1$}{
 let $u$ be the neighbor of $v$ in \gab\;
 \aset\ $y(v)=0$;\quad \aset\ $y(u)=1$\;
}
\nlset{2}\ForAll{$P=uv_{1}v_{2}w \subset G'$ \unproc, $\deg{u}\geq3, \deg{v_{1}}=\deg{v_{2}}=2$}{
 \aset\ $y(u)=y(w)=y(v_{1}) = 1$\;
 \aset\ $y(v_{2})=0$\;
}
\nlset{3}\ForAll{$v \in V'$ \unproc, $\deg{v}=2 \wedge \exists u\in\N(v), \deg{u}\geq3$}{
\nlset{3.1} \ElseIf{$u$ \unproc, $w$ \proc}{
 \aset\ $y(v)=0$;\quad
 \aset\ $y(u)=1$\;
 }
\nlset{3.2} \ElseIf{both $u,w$ \unproc}{
 \aset\ $y(v)=0$;\quad
 \aset\ $y(u)=y(w)=1$\tcc*{with $x(u)\geq\frac{1}{2}$ and $x(w)\geq\frac{1}{2}$}
 }
\nlset{3.3} \If{both $u,w$ \proc}{
 \aset\ $y(v)=0$\;
 }
\nlset{3.4} \ElseIf{$u$ \proc, $w$ \unproc}{
 \aset\ $y(v)=0$;\quad
 \aset\ $y(w)=1$\tcc*{$y(u)=1$ already set and $x(w)\geq\frac{1}{2}$}
 }
}
\nlset{4}\ForAll{$v \in V'$ \unproc, $\deg{v}=2$}{
\nlset{4.1}\ElseIf{$u$ \unproc, $w$ \proc}{
 \aset\ $y(v)=0$;\quad
 \aset\ $y(u)=1$\;
 }
\nlset{4.2}\ElseIf{both $u,w$ \unproc}{
 \aset\ $y(v)=0$;\quad
 \aset\ $y(u)=y(w)=1$\tcc*{with $x(u)\geq\frac{1}{2}$ and $x(w)\geq\frac{1}{2}$}
 }
\nlset{4.3}\If{both $u,w$ \proc}{
 \aset\ $y(v)=0$\;
 }
\nlset{4.4}\ElseIf{$u$ \proc, $w$ \unproc}{
 \aset\ $y(v)=0$;\quad
 \aset\ $y(w)=1$\tcc*{$y(u)=1$ already set and $x(w)\geq\frac{1}{2}$}
 }
}
\nlset{5}\ForAll{$v \in V$}{
 \If{$x(v)=\frac{1}{2}$}{
  \aset\ $y(v)= 1$\tcc*{$y(v)= \min\{1,2\cdot x(v)\}$}
 }
}
\caption{\textsc{Deterministic Rounding}}
\label{alg:rounding}
\end{algorithm}

An analysis of the algorithm is provided by the following \autoref{lem:yInteger} and \autoref{lem:yThreeHalf}.

\begin{lemma}\label{lem:yInteger}
The assignment $y$ generated by \Autoref{alg:rounding} is an integer solution and satisfies $y(u)=1$ for all $u \in V'$ with $\deg{u} \geq 3$. 
 
\end{lemma}
\begin{proof}
 Any high-degree neighbor of degree-1 vertices is set to $1$ in step \astep{1} of the algorithm.

 Since either step \astep{3} or \astep{4} is processing every single degree-2 vertex $v \in V$ with $x(v)=\frac{1}{2}$, there are no leftover vertices $v\in V'$ of degree $2$ with fractional values.

 Assume that there is a vertex $u \in V', \deg{u}\geq3$ and $x(u)=y(u)=\frac{1}{2}$. Then $u$ has at least one degree 2 neighbor $v_1$ with $x(v_1)=\frac{1}{2}$. Because of step \astep{3} and \astep{4} of the algorithm, $v_1$ must have been processed by another degree 2 vertex $v_2$, setting $y(v_1)=1$. This again introduces another neighbor $w$ of $v_2$ with $y(w)=1$ and leads to the situation of a path $uv_{1}v_{2}w$ described in step \astep{2}. In this case, the algorithm sets $y(u)=1$ and thus we have a contradiction to the above assumption.
\end{proof}

\begin{lemma}\label{lem:yThreeHalf}
 The assignment $y$ generated by \Autoref{alg:rounding} satisfies $y(\Vgood)\leq \frac{3}{2}\cdot x(\Vgood)$.
\end{lemma}
\begin{proof}

 The algorithm partitions the graph induced by $\Vgood$ into edge-disjoint subgraphs, namely stars whose leaves are degree-1 vertices and paths of length $\leq 4$ whose internal nodes are degree-2 vertices. We show that for each such subgraph $P_i$, $y(P_i)\leq \frac{3}{2} \cdot x(P_i)$ and furthermore $y(v)=1$ for each $v\in\Vgood$ which is contained in more than one such subgraph. 

 In step \astep{1} of the algorithm all degree-1 vertices and their neighbors are processed.

 In step \astep{2} the subgraphs are \unproc\ paths $P_i$ of length $3$. Since $P_i= \squartet$ contains two disjoint edges $\{u,v_1\}, \{v_2,w\}$, $x(P_i)\geq 2$ and particularly $x(v_2) + x(w) \geq 1$. Therefore $y(P_i)=3 \leq \frac{3}{2}\cdot x(P_i)$ holds via mapping \quartet (where the gray color indicates a \proc\ vertex) and $y$ restricted to $P_i$ (denoted as $y\restrict P_i$) is a vertex cover for $P_i$.

 In step \astep{3} all paths $P_i=\striplet$ are processed, where at least one of $u,w$ is of degree $\geq3$. In cases \astep{3.1}-\astep{3.4} the algorithm considers all possible combinations of some of these nodes being already processed.

 In case \astep{3.1} $u$ is marked \unproc, $w$ is already \proc\ and $x(u)\geq\frac{1}{2}$. The rounding algorithm sets $y(v)=0$ and $y(u)=1$, mapping \triplet{vertex}{\geq\negthickspace\nicefrac{1}{2}}{vertex2}{1}, again yielding a vertex cover $y\restrict P_i$ for $P_i$ with $y(P_i) \leq x(P_i)$.

 In case \astep{3.2} we have that both $u,w$ are marked as \unproc\ and since $x(v)=\frac{1}{2}$ we have that $x(u)\geq\frac{1}{2}$ and $x(w)\geq\frac{1}{2}$. The rounding algorithm sets $y(v)=0, y(u)=y(w)=1$, mapping \triplet{vertex}{\geq\negthickspace\nicefrac{1}{2}}{vertex}{\geq\negthickspace\nicefrac{1}{2}}, and since $x(u)\geq\frac{1}{2}$ and $x(w)\geq\frac{1}{2}$ we have that $y(P_i) \leq \frac{4}{3}\cdot x(P_i)$.

 In case \astep{3.3} both $u,w$ are marked as \proc\ and therefore $y(u)=y(w)=1$, since $u,w$ are adjacent to processed degree one or degree two vertices other than $v$. The algorithm sets $y(v)=0$, mapping \triplet{vertex2}{1}{vertex2}{1}. Hence $y \restrict P_i$ is a vertex cover for $P_i$ with $y(P_i) \leq x(P_i)$.

 In case \astep{3.4} $u$ is already \proc\ and $w$ is still marked \unproc. Since $x(v)=\frac{1}{2}$ we have that $x(w)\geq \frac{1}{2}$. The rounding algorithm sets $y(v)=0$ and $y(w)=1$, mapping \triplet{vertex2}{1}{vertex}{\geq\negthickspace\nicefrac{1}{2}}, and since $x(w)\geq\frac{1}{2}$ it yields a vertex cover $y\restrict P_i$ for $P_i$ with $y(P_i) \leq x(P_i)$.

 Step \astep{4} considers all remaining unprocessed vertices of degree 2. If $v$ is such a vertex with neighborhood $\N(v)=\{u,w\}$, the sub-cases \astep{4.1}-\astep{4.4} are treated analogously to cases \astep{3.1}-\astep{3.4} and the mapping $x\mapsto y$ achieves $y(P_i) \leq \frac{4}{3}\cdot x(P_i)$ on the considered paths $P_i$.

 After steps \astep{0}-\astep{4} of the algorithm there may still be some remaining high-degree vertices $u \in \Vgood, \deg{u}\geq 3$ with $x(u)=y(u)=\frac{1}{2}$.
 These are treated separately (and rounded to $y(u)=1$ together with all other vertices in $V\setminus (V'\setminus \Vgood)$) in step \astep{5} of the algorithm. We have to argue that $y(\Vgood) \leq \frac{3}{2}\cdot x(\Vgood)$ still holds true.

 We consider first the case that $u\in V', \deg{u}\geq 3$ and $x(u)=y(u)=\frac{1}{2}$. Then $u$ has a neighbor $v$ of degree $\leq 2$ with $x(v)=\frac{1}{2}$ and $y(v)=1$, and since $y(u)=\frac{1}{2}$ we have $\deg{v}=2$. Let $v_2$ be the other neighbor of $v$, then $\deg{v_2}=1$ (since otherwise the second neighbor $w$ of $v_2$ would give rise to a path of length $3$, containing also $u$ and hence would have been processed in step \astep{2}). But then locally on the set $\{u,v,v_2\}$ we have the mapping \zweisprung\ with a local ratio of $\frac{4}{3}$.

 Let us now assume $u \in \Vgood\setminus V', \deg{u}\geq 3$ and $x(u)=y(u)=\frac{1}{2}$. Then every degree-2 neighbor $v$ has $x(v)\neq \frac{1}{2}$, hence $x(v)=1$, and therefore $y(v)=1$. We show that $v \notin V'$, i.e. that $v$ was not processed by the algorithm and can be treated as a part of a subgraph disjoint to $G'$ in $\gab$. 
 Let $w\in\N(v)$ be the second neighbor of $v$ besides $u$. Then $x(w)=0$ since otherwise (in case $x(w)\geq \frac{1}{2}$) we could decrease $x(v)$ from $1$ to $\frac{1}{2}$ and still have a feasible half-integral solution, which would contradict the optimality of $x$.
 Therefore $v,w\notin V'$, which means that $v,w$ are not processed by the algorithm. Rounding $y(u)=1$, mapping \doublet, yields a vertex cover $y \restrict \{u,v\}$ with $y(\{u,v\})\leq \frac{4}{3}\cdot x(\{u,v\})$.

 We conclude that the assignment $y: V\mapsto \{0,1\}$ is a vertex cover of \gab\ with $y(\Vgood)\leq\frac{3}{2}\cdot x(\Vgood)$ and $y(V\setminus\Vgood) \leq 2 \cdot x(V\setminus\Vgood)$.
\end{proof}

\subsection{Expected Approximation Ratio}\label{subsec:bounds}

The following lemma shows how to retrieve an expected approximation ratio for our algorithm for \MVC\ in \gab.

\begin{lemma}\label{lem:ratio_lemma}
 If the rounding scheme $x \mapsto y$ satisfies $y(\Vgood) \leq \frac{3}{2} \cdot x(\Vgood)$ and $y(V \setminus \Vgood) \leq 2 \cdot x(V\setminus\Vgood)$ then this gives an approximation ratio \[\frac{y(V)}{\OPT} \leq \frac{y(V)}{x(V)} \leq \frac{x(\Vgood)}{x(V)} \cdot \frac{3}{2} + \frac{x(V\setminus\Vgood)}{x(V)} \cdot 2.\]
\end{lemma}

In order to apply \autoref{lem:ratio_lemma} and to derive an expected approximation ratio for the algorithm, in the following we will give a lower bound on $\E[x(\Vgood)]$ and an upper bound on $x(V)$. The next lemma provides a lower bound on $x(\Vgood)$ in terms of the number of high-degree vertices adjacent to degree-1 and degree-2 nodes.

\begin{lemma}\label{lem:lemma5}
 Let $\gab[\Vgood]$ be the subgraph of \gab\ induced by $\Vgood$.
 For every optimal half-integral solution $x$ for the \MVC\ \LP, the size of the half-integral solution restricted to $\Vgood$ is lower-bounded by the size of the high-degree neighborhood of degree-1 and degree-2 vertices:
 \[x(\Vgood)\geq \frac{1}{2} \cdot \big\lvert\left\{u \in V \vert \deg{u}\geq 3 \wedge \exists v \in \N(u), \deg{v}\in\{1,2\} \right\}\big\rvert\]
\end{lemma}
\begin{proof}
 Let $\Vgood = X \cup Y, X=\{v \in V \vert \deg{v}\in\{1,2\}\}$ and $Y=\{u \in V \vert \deg{u}\geq 3 \wedge \exists v \in \N(u), \deg{v}\in\{1,2\}\}$. Choose some arbitrary function $f:Y \to E(X,Y)$ such that for every $u \in Y, f(u)=\{u,v\}$ for some $v\in X$ adjacent to $u$. $f(Y)$ consists of pairwise disjoint paths $Q_{1},\dots,Q_{m}$ of length $\leq 2$, such that each path contains one or two vertices from $Y$. This implies  $x(\Vgood)\geq m \geq \frac{\lvert Y \rvert}{2}$.
\end{proof}

\subsubsection{First Analysis}\label{subsec:FirstAnalysis}

We will now estimate the expected number of high-degree vertices adjacent to vertices of degree one or two, which -- combined with the preceding \autoref{lem:lemma5} -- gives a lower bound on $\E[x(\Vgood)]$. We prove the following theorem:

\begin{theorem}\label{thm:ExVgood}
 \begin{align}
  \E[x(\Vgood)] &\geq \frac{1}{2}\cdot\E\left[\big\lvert\left\{u \in V \vert \deg{u}\geq 3 \wedge \exists v \in \N(u), \deg{v}\in\{1,2\} \right\}\big\rvert\right]\nonumber\\
  &= \frac{1}{2}\;\cdot\; \sum_{\mathclap{u:\deg{u}\geq 3}} \eta(u)\label{eqn:ExVgoodEta}\\
  &\geq \frac{\e^{\alpha}}{2^{\beta}}\cdot\frac{\zeta(\beta)-1-\frac{1}{2^{\beta}}}{\zeta(\beta-1)}\label{eqn:ExVgood}
 \end{align}
 where $\eta(u)$ is the probability that $u\in V$ has a neighbor in the set of vertices of degree one or two.
\end{theorem}

In order to provide bounds on the probability $\eta(u)$ for a vertex $u$ of degree $d$ of having a degree-1 or degree-2 neighbor, we consider how edges are generated in the random matching procedure of the distribution $M(\alpha,\beta)$: $\deg{u}$ copies of $u$ are randomly matched with the copies of the remaining vertices $v \in V, v\neq u$. We use the following lower bound on $\eta(u)$.

\begin{lemma}\label{lem:etaLemma}
 For every $u$ with $\deg{u}\geq 3$, $\eta(u)\geq \frac{1}{2^{\beta-1}\cdot\sum_{i=1}^{\dmax} \frac{1}{i^{\beta-1}}}$.
\end{lemma}
\begin{proof}
 \begin{align*}
  \eta(u) &\geq \Pr(\text{the first copy of $u$ is neighbor of a degree-2-node})\\
  &= \frac{2\cdot \#\text{deg-2-nodes}}{\left(\sum_{v\in V}\deg{v}\right)-1}\\
  &\geq \frac{2\cdot\frac{\e^{\alpha}}{2^{\beta}}}{\sum_{i=1}^{\dmax}i\cdot\frac{\e^{\alpha}}{i^{\beta}}} = \frac{\frac{1}{2^{\beta-1}}}{\sum_{i=1}^{\dmax}\frac{1}{i^{\beta-1}}}, \text{ where } \dmax=\e^{\frac{\alpha}{\beta}} \text{ is the maximum degree of }\gab.
 \end{align*}
\end{proof}

In Equation~\ref{eqn:ExVgoodEta} we substitute $\eta(u)$ by the bound given in \autoref{lem:etaLemma} and obtain:
\begin{align}
 \E[x(\Vgood)] &\geq \frac{1}{2}\;\cdot\;\sum_{\mathclap{u:\deg{u}\geq 3}}\eta(u)\nonumber\\
 &= \frac{1}{2}\cdot\left(\sum_{i=1}^{\dmax}\frac{\e^{\alpha}}{i^{\beta}}-\e^{\alpha}-\frac{\e^{\alpha}}{2^{\beta}}\right)\cdot\frac{1}{2^{\beta-1}\cdot\sum_{i=1}^{\dmax} \frac{1}{i^{\beta-1}}}\nonumber\\
 &= \frac{\e^{\alpha}}{2^{\beta}}\cdot\frac{\sum_{i=1}^{\dmax}\frac{1}{i^{\beta}}-1-\frac{1}{2^{\beta}}}{\sum_{i=1}^{\dmax}\frac{1}{i^{\beta-1}}}\label{eqn:ExVgoodSum}
\end{align}

We will now show that in Inequality~\ref{eqn:ExVgoodSum} we can replace the terns $\sum_{i=1}^{\dmax}\frac{1}{i^{\beta}}$ and $\sum_{i=1}^{\dmax}\frac{1}{i^{\beta-1}}$ by $\zeta(\beta)$ and $\zeta(\beta-1)$, respectively. We make use of the following lemma.

\begin{lemma}\label{lem:AB_geq_ab}
 For $A,B,a,b > 0$, $\frac{A}{B}\geq\frac{A+a}{B+b} \;\Longleftrightarrow\; \frac{A}{B}\geq\frac{a}{b}$.
\end{lemma}

Therefore, in order to show 
\[\E[x(\Vgood)] \geq \frac{e^{\alpha}}{2^{\beta}}\cdot\frac{\zeta(\beta)-1-\frac{1}{2^{\beta}}}{\zeta(\beta-1)},\]
it is sufficient to show that there exists a $\dmax_{0}$ such that for all $\dmax\geq \dmax_{0}$ the following holds
\[\frac{\sum_{i=1}^{\dmax}\frac{1}{i^{\beta}}-1-\frac{1}{2^{\beta}}}{\sum_{i=1}^{\dmax}\frac{1}{i^{\beta-1}}} \geq \frac{\frac{1}{(\dmax+1)^{\beta}}}{\frac{1}{(\dmax+1)^{\beta-1}}} = \frac{1}{\dmax+1}.\]

This is provided by the following lemma.
\begin{lemma}
 There exists a $\dmax_{0}\geq 8$ such that for all $\dmax\geq\dmax_{0}$, $\frac{\sum_{i=1}^{\dmax}\frac{1}{i^{\beta}}-1-\frac{1}{2^{\beta}}}{\sum_{i=1}^{\dmax}\frac{1}{i^{\beta-1}}} \geq \frac{1}{\dmax+1}$.
\end{lemma}
\begin{proof}
 The above inequality is equivalent to
 \begin{align}
  &\phantom{\Longleftrightarrow}& \sum_{i=1}^{\dmax}\frac{1}{i^{\beta}}-1-\frac{1}{2^{\beta}} &\geq \sum_{i=1}^{\dmax}\frac{1}{\dmax +1}\cdot\frac{1}{i^{\beta-1}}\nonumber\\
  &\Longleftrightarrow& \sum_{i=1}^{\dmax}\left(\frac{1}{i^{\beta}}-\frac{1}{\dmax+1}\cdot\frac{1}{i^{\beta-1}}\right) &\geq 1+ \frac{1}{2^{\beta}}\nonumber\\
  &\Longleftrightarrow& \sum_{i=1}^{\dmax}\frac{\dmax+1-i}{(\dmax+1)i^{\beta}} &\geq 1+ \frac{1}{2^{\beta}}\label{eqn:ExVgoodSum3}
 \end{align}
Suppose $\dmax\geq 8$, then the sum on the left-hand side of the Inequality~\ref{eqn:ExVgoodSum3} is bounded by the sum of the terms with indices $i=1,2,4,8$:
 \begin{multline}
  \sum_{i=1}^{\dmax}\frac{\dmax+1-i}{(\dmax+1)i^{\beta}} \geq \frac{\dmax}{\dmax+1} + \frac{\dmax-1}{(\dmax+1)2^{\beta}} + \frac{\dmax-3}{(\dmax+1)4^{\beta}} + \frac{\dmax-7}{(\dmax+1)8^{\beta}}\\
   = \frac{\dmax 8^{\beta} + (\dmax-1)4^{\beta} + (\dmax-3)2^{\beta} + \dmax-7}{(\dmax+1)8^{\beta}}.\label{eqn:temp}
 \end{multline}
Using Inequality~\ref{eqn:temp} and the fact that $1+\frac{1}{2^{\beta}} = \frac{(\dmax+1)8^{\beta} + (\dmax+1)4^{\beta}}{(\dmax+1)8^{\beta}}$, in order to prove Inequality~\ref{eqn:ExVgoodSum3} it is sufficient to show the following:
 \begin{align*}
  &\phantom{\Longleftrightarrow}& \frac{\dmax 8^{\beta} + (\dmax-1)4^{\beta} + (\dmax-3)2^{\beta} + \dmax-7}{(\dmax+1)8^{\beta}} &\stackrel{!}{\geq} \frac{(\dmax+1)8^{\beta} + (\dmax+1)4^{\beta}}{(\dmax+1)8^{\beta}}\\
  &\Longleftrightarrow& (\dmax-3)2^{\beta} + \dmax-7 &\stackrel{!}{\geq} 8^{\beta} + 2\cdot 4^{\beta}.
 \end{align*}
 This is valid for $\dmax\geq\frac{8^{\beta}+2\cdot4^{\beta}+6\cdot2^{\beta}+7}{1+2^{\beta}}$. Hence we choose $\dmax_{0}=\ceil*{\frac{8^{\beta}+2\cdot4^{\beta}+6\cdot2^{\beta}+7}{1+2^{\beta}}}$.
\end{proof}
This completes the proof of \autoref{thm:ExVgood}. The next lemma provides an upper bound for $x(V)$:
\begin{lemma}\label{lem:UB_V}
 $x(V) \leq \frac{1}{2}\zeta(\beta)\e^{\alpha}$
\end{lemma}
\begin{proof}
 In order to get an upper bound for $x(V)$ we construct a feasible half-integral solution for \gab\ by setting $x(v)=\frac{1}{2}$ for all $v\in V$ where $\frac{1}{2}\sum_{v\in V}\leq \frac{1}{2}\zeta(\beta)\e^{\alpha}$.
\end{proof}

Now let us restate the main \autoref{thm:main} and finish the proof.

\begin{theorem*}
 For $\beta > 2$ the \mvc\ problem in $(\alpha,\beta)$-Power Law Graphs \gab\ can be approximated with expected approximation ratio $\rho \leq \mainratio$.
\end{theorem*}
\begin{proof}
 \Autoref{alg:rounding} achieves an approximation ratio of $\frac{3}{2}$ for \MVC\ in the subgraph induced by \Vgood\ in \gab\ and a ratio of $2$ in $\gab[V\setminus\Vgood]$, i.e.
 \[\rho\leq\E\left[\frac{3}{2}\cdot\frac{x(\Vgood)}{x(V)} + 2\cdot\frac{x(V)-x(\Vgood)}{x(V)} \right] = \E\left[2-\frac{1}{2}\cdot\frac{x(\Vgood)}{x(V)}\right].\] From \autoref{thm:ExVgood} and \autoref{lem:UB_V} we have that $\E[x(\Vgood)] \geq \frac{1}{2} \cdot \frac{\left( \zeta(\beta) - 1 - \frac{1}{2^{\beta}}\right) \e^{\alpha}}{2^{\beta -1}\zeta(\beta-1)}$ and $x(V) \leq \frac{1}{2}\cdot\zeta(\beta)\e^{\alpha}$.
 This yields
 \[\E\left[\frac{x(\Vgood)}{x(V)}\right]\geq\frac{\frac{1}{2} \cdot \frac{\left( \zeta(\beta) - 1 - \frac{1}{2^{\beta}}\right) \e^{\alpha}}{2^{\beta -1}\zeta(\beta-1)}}{\frac{1}{2}\cdot\zeta(\beta)\e^{\alpha}}
 = \frac{\zeta(\beta)-1-\frac{1}{2^{\beta}}}{2^{\beta-1}\zeta(\beta-1)\zeta(\beta)}\]
and
 \[\rho\leq 2-\frac{1}{2}\cdot\frac{\zeta(\beta)-1-\frac{1}{2^{\beta}}}{2^{\beta-1}\zeta(\beta-1)\zeta(\beta)}
 = 2-\frac{\zeta(\beta)-1-\frac{1}{2^{\beta}}}{2^{\beta}\zeta(\beta-1)\zeta(\beta)}\]
\end{proof}

\subsubsection{Refined Analysis for $\beta > 2.424$}\label{subsec:SecondAnalysis}

We will now refine the analysis of \Autoref{alg:rounding} by giving a better estimate on the probability $\eta(u, U)$ of a high-degree node $u$ being adjacent to a vertex in the set $U$, i.e. a vertex of degree one or two. However, this analysis will only apply to the more restricted range of $\beta > 2.424$. Again, we will first obtain a bound on the expected approximation ratio of the algorithm in terms of the partial sums $\sum_{i=1}^{\Delta}\frac{1}{i^\beta}$ and $ \sum_{i=1}^{\Delta}\frac{1}{i^\beta-1}$ and then show that these can be replaced by $\zeta (\beta)$ and $\zeta (\beta-1)$, respectively.

\begin{lemma}
 For every $u$ with $\deg{u}\geq 3$ and $U\subseteq V$,
 \[\eta(u,U)\geq \etau.\]
\end{lemma}
\begin{proof}
For a given set $U$ of vertices from $\gab$ we let $\deg{U}=\sum_{v\in U}\deg{v}$. Furthermore let $\eta (u,U)$ be the probability that $u$ is connected to at least one node in $U$. We obtain
 \begin{align*}
  \eta(u,U) &= \Pr(u \text{ matches to } U)\\
  &= \sum_{j=1}^{\deg{u}} \Pr(j\text{-th copy is first one matching to } U)\\
  &= \sum_{j=1}^{\deg{u}} \frac{\deg{u}}{\sum_{i = 1}\frac{e^{\alpha}}{i^{\beta-1}}-(j-1)-1}\prod_{k=1}^{j-1}\left( 1- \frac{\deg{U}}{\sum_{i = 1}\frac{e^{\alpha}}{i^{\beta-1}}-1-(k-1)} \right)
 \end{align*}
Now define $N=\sum_{i = 1}\frac{e^{\alpha}}{i^{\beta-1}}$. We have:
 \begin{align*}
  \hspace*{-5ex}\eta(u,U) &= \sum_{j=1}^{\deg{u}} \frac{\deg{U}}{N-j}\prod_{k=1}^{j-1}\frac{N-\deg{U}-k}{N-k}
  \hspace*{10ex}\geq  \sum_{j=1}^{\deg{u}} \frac{\deg{U}}{N-j}\left(\frac{N-\deg{U}-j+1}{N-j+1}\right)^{j-1}\\
  &\geq  \sum_{j=1}^{\deg{u}} \frac{\deg{U}}{N-j}\left(\frac{N-\deg{U}-\deg{u}+1}{N-\deg{u}+1}\right)^{j-1}
  \geq  \sum_{j=1}^{\deg{u}} \frac{\deg{U}}{N}\left(\frac{N-\deg{U}-\deg{u}+1}{N-\deg{u}+1}\right)^{j-1}\\
  &= \frac{\deg{U}}{N}\left[\frac{1-\left(\frac{N-\deg{U}-\deg{u}+1}{N-\deg{u}+1} \right)^{\deg{u}}}{1-\frac{N-\deg{U}-\deg{u}+1}{N-\deg{u}+1}}\right]\\
  &= \frac{\deg{U}}{N}\left[1-\left(\frac{N-\deg{U}-\deg{u}+1}{N-\deg{u}+1} \right)^{\deg{u}}\right] \cdot \frac{N-\deg{u}+1}{\deg{U}}\\
  &= \frac{N-\deg{u}+1}{N}\left[1-\left(\frac{N-\deg{U}-\deg{u}+1}{N-\deg{u}+1}\right)^{\deg{u}}\right]
 \end{align*}
Since the function $\left(\frac{N-\deg{U}-\deg{u}+1}{N-\deg{u}+1}\right)^{\deg{u}}$ is monotone decreasing in $\deg{u}$ it follows that:
 \begin{align*}
  \eta(u,U) &\geq \frac{N-\dmax+1}{N}\left[1-\left(\frac{N-\deg{U}-3+1}{N-3+1}\right)^{3}\right]\\
  &= \frac{\sum_{i = 1}\frac{e^{\alpha}}{i^{\beta-1}}-e^{\frac{\alpha}{\beta}}+1}{\sum_{i = 1}\frac{e^{\alpha}}{i^{\beta-1}}}\left[1-\left(\frac{\sum_{i = 1}\frac{e^{\alpha}}{i^{\beta-1}}-\deg{U}-3+1}{\sum_{i = 1}\frac{e^{\alpha}}{i^{\beta-1}}-3+1}\right)^{3}\right]
 \end{align*}
\end{proof}

Because of Equation~\ref{eqn:ExVgoodEta} we have ${\displaystyle\E[x(\Vgood)] \geq \frac{1}{2}\;\cdot\;\sum_{\mathclap{u:\deg{u}\geq 3}}\eta(u,U)}$ and we obtain the following approximation ratio:
\begin{align}
 \hspace*{-1ex}\rho &\leq \E\left[2-\frac{1}{2}\cdot\frac{x(\Vgood)}{x(V)}\right]\nonumber\\
 &\leq 2-\frac{1}{2}\cdot \frac{\left(\sum_{i=1}^{\dmax}\frac{\e^{\alpha}}{i^{\beta}}-\e^{\alpha}-\frac{\e^{\alpha}}{2^{\beta}}\right)\cdot\etau}{\frac{1}{2}\sum_{i=1}^{\dmax}\frac{\e^{\alpha}}{i^{\beta}}}\nonumber\\
 &= 2- \underbrace{\frac{\left(\sum_{i=1}^{\dmax}\frac{1}{i^{\beta}}-1-\frac{1}{2^{\beta}}\right)	 \cdot	 \left(\sum_{i = 1}^{\dmax} \frac{1}{i^{\beta-1}}-\frac{\dmax}{\e^{\alpha}} + \frac{1}{\e^{\alpha}}\right)}{\left(\sum_{i = 1}^{\dmax}\frac{1}{i^{\beta-1}}\right) 	\cdot 	\left(\sum_{i=1}^{\dmax}\frac{1}{i^{\beta}}\right)}}_{F}
 \left[1-{\underbrace{\left(\frac{\sum_{i = 1}^{\dmax}\frac{1}{i^{\beta-1}}-\frac{\deg{v}}{\e^{\alpha}}-\frac{2}{\e^{\alpha}}}{\sum_{i = 1}^{\dmax}\frac{1}{i^{\beta-1}}-\frac{2}{\e^{\alpha}}}\right)}_{C}}^{3}\right]\label{eqn:Erho}
\end{align}

Now we show that, in Inequality~\ref{eqn:Erho}, we can replace the partial sums $\sum_{i = 1}^{\dmax}\frac{1}{i^{\beta}}$ and $\sum_{i = 1}^{\dmax}\frac{1}{i^{\beta-1}}$ by $\zeta(\beta)$ and $\zeta(\beta-1)$ respectively.
First, we consider the term $C$ where $\deg{v}=\e^{\alpha}\left(1+\frac{1}{2^{\beta-1}}\right)$
, i.e. the number of copies of degree-1 and degree-2 vertices:
\[C=\frac{\sum_{i = 1}^{\dmax}\frac{1}{i^{\beta-1}}-\frac{\e^{\alpha}\left(1+\frac{1}{2^{\beta-1}}\right)}{\e^{\alpha}}-\frac{2}{\e^{\alpha}}}{\sum_{i = 1}^{\dmax}\frac{1}{i^{\beta-1}}-\frac{2}{\e^{\alpha}}}
 = \frac{\sum_{i = 1}^{\dmax}\frac{1}{i^{\beta-1}}-\left(1+\frac{1}{2^{\beta-1}}\right)-\frac{2}{\dmax^{\beta}}}{\sum_{i = 1}^{\dmax}\frac{1}{i^{\beta-1}}-\frac{2}{\dmax^{\beta}}}\]

We show that following inequality holds true:
 \begin{align*}
  &\phantom{\Longleftrightarrow}& \frac{\sum_{i = 1}^{\dmax}\frac{1}{i^{\beta-1}}-\left(1+\frac{1}{2^{\beta-1}}\right)-\frac{2}{\dmax^{\beta}}}{\sum_{i = 1}^{\dmax}\frac{1}{i^{\beta-1}}-\frac{2}{\dmax^{\beta}}} &\leq \frac{\frac{1}{(\dmax+1)^{\beta-1}} + \frac{2}{\dmax^{\beta}} - \frac{2}{(\dmax+1)^{\beta}}}{\frac{1}{(\dmax+1)^{\beta-1}} + \frac{2}{\dmax^{\beta}} - \frac{2}{(\dmax+1)^{\beta}}}\\
  &\Longleftrightarrow& \frac{\sum_{i = 1}^{\dmax}\frac{1}{i^{\beta-1}}-\left(1+\frac{1}{2^{\beta-1}}\right)-\frac{2}{\dmax^{\beta}}}{\sum_{i = 1}^{\dmax}\frac{1}{i^{\beta-1}}-\frac{2}{\dmax^{\beta}}} &\leq 1\\
  &\Longleftrightarrow& \sum_{i = 1}^{\dmax}\frac{1}{i^{\beta-1}}-\left(1+\frac{1}{2^{\beta-1}}\right)-\frac{2}{\dmax^{\beta}} &\leq \sum_{i = 1}^{\dmax}\frac{1}{i^{\beta-1}}-\frac{2}{\dmax^{\beta}}\\
  &\Longleftrightarrow& \sum_{i = 1}^{\dmax}\frac{1}{i^{\beta-1}}-\left(1+\frac{1}{2^{\beta-1}}\right) &\leq \sum_{i = 1}^{\dmax}\frac{1}{i^{\beta-1}} \qquad\square
 \end{align*}

We have
\[F=\frac{\left(\sum_{i=1}^{\dmax}\frac{1}{i^{\beta}}-1-\frac{1}{2^{\beta}}\right)	 \cdot	 \left(\sum_{i = 1}^{\dmax} \frac{1}{i^{\beta-1}}-\frac{\dmax}{\e^{\alpha}} + \frac{1}{\e^{\alpha}}\right)}	{\left(\sum_{i = 1}^{\dmax}\frac{1}{i^{\beta-1}}\right) 	\cdot 	\left(\sum_{i=1}^{\dmax}\frac{1}{i^{\beta}}\right)}.\]
We let $S_{\beta}=\sum_{i=1}^{\dmax}\frac{1}{i^{\beta}}$ and $S_{\beta-1}=\sum_{i=1}^{\dmax}\frac{1}{i^{\beta-1}}$ and recall that $\e^{\alpha}=\dmax^{\beta}$. According to \autoref{lem:AB_geq_ab} it remains to show the following inequality:
\begin{multline*}
 \frac{\left({S_{\beta}}-1-\frac{1}{2^{\beta}}\right)	 \cdot	 \left(S_{\beta-1}-\frac{1}{\dmax^{\beta-1}} + \frac{1}{\dmax^{\beta}}\right)}	{\left(S_{\beta-1}\right) 	\cdot 	\left(S_{\beta}\right)}\\
 \geq
 \frac{1} {\frac{1}{(\dmax+1)^{\beta}}\left(S_{\beta-1}\right) 	+ 	\frac{1}{(\dmax+1)^{\beta-1}}\left(S_{\beta} + \frac{1}{(\dmax+1)^{\beta}}\right)} \cdot \left[\frac{1}{(\dmax+1)^{\beta}}\left(S_{\beta-1}-\frac{1}{\dmax^{\beta-1}} + \frac{1}{\dmax^{\beta}}\right)\right.\\
 + \left.\left(S_{\beta}-1-\frac{1}{2^{\beta}}+\frac{1}{(\dmax+1)^{\beta}}\right) \cdot \left(-\frac{1}{(\dmax+1)^{\beta}}+\frac{1}{(\dmax+1)^{\beta-1}}+\frac{1}{(\dmax)^{\beta-1}}-\frac{1}{(\dmax)^{\beta}}\right)\right]
\end{multline*}
which is equivalent to
\begin{multline*}
 \left(S_{\beta}-1-\frac{1}{2^{\beta}}\right) \left(S_{\beta-1}-\frac{\dmax-1}{\dmax^{\beta}}\right) \left[\frac{1}{(\dmax+1)^{\beta}}S_{\beta-1} + \left(S_{\beta}+\frac{1}{(\dmax+1)^{\beta}}\right)\frac{1}{(\dmax+1)^{\beta-1}}\right]\\
 \geq
 \left[\frac{1}{(\dmax+1)^{\beta}}\left(S_{\beta-1}+\frac{\dmax-1}{\dmax^{\beta}}\right)\cdot S_{\beta}\cdot S_{\beta-1}\right.\\
 + \left.\left(S_{\beta}-1-\frac{1}{2^{\beta}}+\frac{1}{(\dmax+1)^{\beta}}\right) \cdot \left(\frac{\dmax}{(\dmax+1)^{\beta}} + \frac{\dmax-1}{\dmax^{\beta}}\right)\cdot S_{\beta} \cdot S_{\beta-1}\right].
\end{multline*}
We rearrange terms and get
\begin{multline}
 S_{\beta-1}^{2}S_{\beta}\frac{1}{(\dmax+1)^{\beta}} + S_{\beta}^{2}S_{\beta-1}\left(\frac{\dmax}{(\dmax+1)^{\beta}}+\frac{\dmax-1}{\dmax^{\beta}}\right)\\
 + S_{\beta}S_{\beta-1}\left[-\frac{\dmax-1}{\dmax^{\beta}}\cdot\frac{1}{(\dmax+1)^{\beta}} + \left(\frac{\dmax}{(\dmax+1)^{\beta}}+\frac{\dmax-1}{\dmax^{\beta}}\right) \left(-1-\frac{1}{2^{\beta}}+\frac{1}{(\dmax+1)^{\beta}}\right)\right] \\
 \leq
 S_{\beta-1}^{2}S_{\beta}\frac{1}{(\dmax+1)^{\beta}} + S_{\beta-1}S_{\beta}^{2}\frac{1}{(\dmax+1)^{\beta-1}}\\
 + S_{\beta}S_{\beta-1}\left[\frac{1}{(\dmax+1)^{\beta}}\cdot\frac{1}{(\dmax+1)^{\beta-1}} - \frac{\dmax-1}{\dmax^{\beta}} - \left(1+\frac{1}{2^{\beta}}\right)\frac{1}{(\dmax+1)^{\beta-1}}\right]\\
 + S_{\beta}^{2}\frac{1}{(\dmax+1)^{\beta-1}} + S_{\beta-1}^{2}\left(-1-\frac{1}{2^{\beta}}\right)\frac{1}{(\dmax+1)^{\beta}}\\
 + S_{\beta-1}\left[-\left(1+\frac{1}{2^{\beta}}\right)\frac{1}{(\dmax+1)^{\beta}}\cdot\frac{1}{(\dmax+1)^{\beta-1}}\left(1+\frac{1}{2^{\beta}}\right)\frac{\dmax-1}{\dmax^{\beta}}\cdot\frac{1}{(\dmax+1)^{\beta}}\right]\\
 + S_{\beta}\left[\left(-1-\frac{1}{2^{\beta}}\right)\frac{1-\dmax}{\dmax^{\beta}}\cdot\frac{1}{(\dmax+1)^{\beta-1}} - \frac{\dmax-1}{\dmax^{\beta}}\cdot\frac{1}{(\dmax+1)^{\beta}}\cdot\frac{1}{(\dmax+1)^{\beta-1}}\right]\\
 + \left(1+\frac{1}{2^{\beta}}\right)\frac{\dmax-1}{\dmax^{\beta}}\cdot\frac{1}{(\dmax+1)^{\beta}}\cdot\frac{1}{(\dmax+1)^{\beta-1}}
\label{eqn:MonsterIneq}
\end{multline}

The following lemma shows that in order to prove Inequality~\ref{eqn:MonsterIneq} it is sufficient to show the respective inequality given by the terms of slowest convergence as $\dmax$ goes to infinity.

\begin{lemma}
 Let $f_{\beta},g_{\beta},F_{\beta},G_{\beta}$ be functions of $\dmax$ depending on the parameter $\beta$ with $|g_{\beta}|,|G_{\beta}|\leq c$ for a constant $c$ depending only on $\beta$. Then $f_{\beta}(\dmax) < F_{\beta}(\dmax)$ for almost all $\dmax$ implies \[f_{\beta}(\dmax)\cdot\frac{1}{\dmax^{\beta-1}} + g_{\beta}(\dmax)\cdot\frac{1}{\dmax^{\beta}} \leq F_{\beta}(\dmax)\cdot\frac{1}{\dmax^{\beta-1}} + G_{\beta}(\dmax)\cdot\frac{1}{\dmax^{\beta}}\] for all but finitely many $\dmax$.
\end{lemma}

Hence it remains to show that
\begin{multline*}
 S_{\beta}^{2}S_{\beta-1}\left(\frac{\dmax}{(\dmax+1)^{\beta}}+\frac{\dmax-1}{\dmax^{\beta}}\right) - S_{\beta}S_{\beta-1}\left(1+\frac{1}{2^{\beta}}\right)\left(\frac{\dmax}{(\dmax+1)^{\beta}} + \frac{\dmax-1}{\dmax^{\beta}}\right)\\
 \leq
 S_{\beta}S_{\beta-1}(-1)\cdot\frac{\dmax-1}{\dmax^{\beta}} + S_{\beta-1}S_{\beta}^{2}\frac{1}{(\dmax+1)^{\beta-1}}
\end{multline*}
which holds true if and only if
\begin{multline*}
 S_{\beta}^{2}S_{\beta-1}\left(\frac{\dmax}{(\dmax+1)^{\beta}}+\frac{\dmax-1}{\dmax^{\beta}}\right)\\
 \leq
 S_{\beta}S_{\beta-1}\left[\left(1+\frac{1}{2^{\beta}}\right)\left(\frac{\dmax}{(\dmax+1)^{\beta}} + \frac{\dmax-1}{\dmax^{\beta}}\right)-\frac{\dmax-1}{\dmax^{\beta}}\right] + S_{\beta-1}S_{\beta}^{2}\frac{1}{(\dmax+1)^{\beta-1}}
\end{multline*}
which can be rewritten as
\begin{align}
 &\phantom{\Longleftrightarrow}& S_{\beta}^{2}S_{\beta-1}\left(\frac{1}{(\dmax+1)^{\beta}}+\frac{\dmax-1}{\dmax^{\beta}}\right) &\leq S_{\beta}S_{\beta-1}\left[\left(1+\frac{1}{2^{\beta}}\right)\frac{\dmax}{(\dmax+1)^{\beta}} + \frac{1}{2^{\beta}}\cdot\frac{\dmax-1}{\dmax^{\beta}}\right]\nonumber\\ 
 &\Longleftrightarrow& S_{\beta}\left(\frac{\dmax-1}{\dmax^{\beta}}\right) &\leq \left(1+\frac{1}{2^{\beta}}\right)\frac{\dmax}{(\dmax+1)^{\beta}} + \frac{1}{2^{\beta}}\cdot\frac{\dmax-1}{\dmax^{\beta}} \nonumber\\
 &\Longleftrightarrow& S_{\beta}(\dmax-1) &\leq \left(1+\frac{1}{2^{\beta}}\right)\frac{\dmax^{\beta+1}}{(\dmax+1)^{\beta}} + \frac{1}{2^{\beta}}\cdot(\dmax-1)\nonumber\\ 
 &\Longleftrightarrow& \left(S_{\beta}-\frac{1}{2^{\beta}}\right)(\dmax-1) &\leq \left(1+\frac{1}{2^{\beta}}\right)\frac{\dmax^{\beta+1}}{(\dmax+1)^{\beta}} \label{eq:2424}
\end{align}
Now Inequality~\ref{eq:2424} follows from the observation that for all $\beta >2.424$, $S_{\beta}-\frac{1}{2^\beta}< 1+\frac{1}{2^\beta}$.

Finally we have shown the following theorem.

\begin{theorem*}
 For all $\beta >2.424$ the \mvc\ problem on $(\alpha,\beta)$-Power Law Graphs \gab\ can be approximated with expected approximation ratio \[\rho \leq  2- \frac{\left(\zeta(\beta)-1-\frac{1}{2^{\beta}}\right)	 \cdot	 \left(\zeta(\beta-1)-\frac{\dmax}{\e^{\alpha}} + \frac{1}{\e^{\alpha}}\right)}{\zeta(\beta-1) 	\cdot 	\zeta(\beta)}
 \left[1-{\left(\frac{\zeta(\beta-1)-\left(1+\frac{1}{2^{\beta-1}}\right)-\frac{2}{\e^{\alpha}}}{\zeta(\beta-1)-\frac{2}{\e^{\alpha}}}\right)}^{3}\right]\]
 This converges to
 \[\rho \leq  2- \frac{\left(\zeta(\beta)-1-\frac{1}{2^{\beta}}\right)	 \cdot	 \zeta(\beta-1)}{\zeta(\beta-1) 	\cdot 	\zeta(\beta)}
 \left[1-{\left(\frac{\zeta(\beta-1)-\left(1+\frac{1}{2^{\beta-1}}\right)}{\zeta(\beta-1)}\right)}^{3}\right]\]
 as $\alpha\to\infty$.
\end{theorem*}

\section{Conclusion}\label{sec:conclusion}

In \autoref{sec:vc_on_plg} we presented a new approximation algorithm for \MVC\ in $(\alpha,\beta)$-\PLG\ with expected approximation ratio of $\rho\leq\mainratio$ in our first analysis of \autoref{subsec:FirstAnalysis}. Moreover, in our refined analysis 
we showed for $\beta > 2.424$ an expected asymptotic approximation ratio of $\rho'\leq\newratio$.
The algorithm itself basically consists of a deterministic rounding procedure on a half-integral solution for \MVC\ (c.f. \Autoref{alg:rounding}).
We showed that this rounding procedure yields an approximation ratio of $\frac{3}{2}$ in the subgraph induced by the low-degree vertices of the $(\alpha,\beta)$-\PLG\ and a 2-approximation in the residual graph.

%
%
%
%

Further research will be directed towards extending the improved analysis also to the range $\beta < 2.424$ and towards investigating the approximability of \MVC\ in other PLG-Models, e.g. the Preferential Attachment Model in \cite{Barabasi1999}.

\section*{Acknowledgements}
The first author is supported by the NRW State within the B-IT Research School. The authors would like to thank Marek Karpinski for helpful remarks and discussions.